\documentclass[11pt,a4paper]{article}
\usepackage[margin=25mm]{geometry}
\usepackage[T1]{fontenc}
\usepackage{lmodern,microtype}
\usepackage{amsmath,amssymb,amsthm,mathtools}
\usepackage{booktabs,longtable}
\usepackage{pgfplots}
\pgfplotsset{compat=1.18}
\usepgfplotslibrary{groupplots}
\usepackage{hyperref}
\hypersetup{colorlinks=true,linkcolor=blue,citecolor=blue,urlcolor=blue}
\newtheorem{theorem}{Theorem}
\newtheorem{lemma}{Lemma}
\newtheorem{proposition}{Proposition}
\newcommand{\sat}{\operatorname{sat}}
\title{Seven Limit Cycles in a Periodically Forced\\Three-Zone Equation}
\author{Chaoyu Hu\thanks{School of Automation, Hangzhou Dianzi University, Hangzhou, Zhejiang 310018, People's Republic of China. Email: \texttt{evaristebernhardwiener@gmail.com}.}
\and Jizheng Chen\thanks{Department of Mathematics, Chang'an University, Xi'an, Shaanxi 710064, People's Republic of China. Email: \texttt{1253183582@qq.com}.}}
\date{}
\begin{document}
\maketitle
\begin{abstract}
We construct a seven-cycle example for a periodically forced scalar
equation with a continuous three-zone piecewise-linear nonlinearity.
Symbolic integration and linear elimination reduce the four-flight
matching conditions to two equations in the residence times. The
competing exterior visits produce an additional turn of the critical
response and lead to the rational parameters
$(A,B,\mu)=(9,3/2,181363/100000)$. At these parameters, finite Taylor
estimates give three sign changes of the return displacement. Scalar
order, analyticity along transverse itineraries, and half-period
reflection then establish at least seven distinct isolated periodic
orbits on the phase cylinder.
\end{abstract}
\noindent\textbf{Keywords:} symbolic computation; piecewise-linear equation;
limit cycle; residence time; grazing.\quad\textbf{MSC 2020:} 34C25, 34A34, 37C27.

\section{The equation and the fixed-parameter result}
Consider
\begin{equation}\label{eq:family}
 \dot x=-Ax+(A+B)\sat(x)+\mu\sin t,\qquad A,B,\mu>0,
\end{equation}
where
\[
 \sat(x)=\begin{cases}-1,&x<-1,\\x,&-1\le x\le1,\\1,&x>1.\end{cases}
\]
A periodic solution gives a closed orbit on the cylinder
$(t\bmod 2\pi,x)$. An isolated such orbit is called a limit cycle.
The exterior vector fields contract perturbations at rate $A$, whereas
the middle vector field expands them at rate $B$. The switching times
determine how these three explicit affine flows compose over a period.

Periodically forced scalar piecewise equations arise in reductions of
three-dimensional piecewise-linear systems. Carmona, Freire, Ponce and
Torres \cite{Carmona2004} studied such a reduction for a two-zone system;
the three-zone saturation family above is treated in \cite{Bravo2025}.
Regularity and the organization of periodic solutions across switching
levels were developed in \cite{Bravo2007}. The surrounding theory of
piecewise-smooth flows is presented in \cite{diBernardo2008}.

Bravo, Carmona, Fern\'andez and Ojeda \cite{Bravo2025} establish
parameter-uniform finiteness for the saturation family and describe
regimes with one, three, and five cycles using a small-slope deformation.
They discuss a conjectural maximum of five. The present example has
strong exterior contraction and gives a larger fixed-parameter count.

Several related results address the number of periodic solutions.
Bravo, Fern\'andez and Ojeda \cite{Bravo2024} bound the Hilbert number
of a two-zone equation with degree-one trigonometric coefficients.
Gasull and Zhao \cite{Gasull2023} obtain a sharp two-cycle bound for
$x'=a(t)+b(t)|x|$ when one coefficient has constant sign.
Bravo and Trinidad-Forte \cite{BravoTrinidad2026} establish finiteness
for multi-zone equations with trigonometric coefficients. Their work on
global centers \cite{Centers2026} characterizes periodic families in a
two-zone equation. Here we construct seven isolated cycles in the
three-zone saturation family.
\begin{theorem}\label{thm:seven}
The equation
\begin{equation}\label{eq:witness}
 \boxed{\dot x=-9x+\frac{21}{2}\sat(x)+\frac{181363}{100000}\sin t}
\end{equation}
has at least seven distinct isolated $2\pi$-periodic orbits on the phase cylinder.
\end{theorem}
This example exceeds the conjectured five-cycle bound for \eqref{eq:family}.

The construction combines symbolic computation with scalar dynamics.
Exact integration gives the flight maps; trigonometric expansion makes
the matching equations linear in the bias and two phase variables.
Determinant elimination then leaves two equations in the residence
times. These formulas locate the parameter example and supply the
expressions used in the rational estimates of the return displacement.

Write $x(t;y)$ for the solution with $x(0;y)=y$, and set
\begin{equation}\label{eq:map}
 P(y)=x(2\pi;y),\qquad D(y)=P(y)-y.
\end{equation}
The vector field is continuous and globally Lipschitz in $x$, so solutions
exist uniquely and depend continuously on $y$. Scalar uniqueness gives
strict order: if $y_1<y_2$, then $x(t;y_1)<x(t;y_2)$ at every common time.
The proof finds three isolated zeros of $D$ above the antisymmetric
solution and uses reflection to obtain three more below it.

\section{Residence times and switching geometry}
For a periodic trajectory with transverse crossings, let $s$ and $v$ be
the times spent below $-1$ and above $1$, respectively, and let
$T=2\pi-s-v$ be its middle-zone residence time. The return multiplier is
\begin{equation}\label{eq:multiplier}
 P'(y)=\exp\bigl(BT-A(s+v)\bigr)
      =\exp\bigl((A+B)T-2\pi A\bigr).
\end{equation}
Indeed, the variational equation has coefficient $B$ in the middle zone
and $-A$ in the exterior zones. Continuity of the vector field makes the
variation continuous at a transverse switching event.
This residence-time form of the multiplier is also recorded in
\cite[Proposition~2.2]{Bravo2025}; differentiability for continuous scalar
piecewise equations is treated in \cite{Bravo2007}.

Raising the initial value lifts the entire trajectory. On a transverse
itinerary with both exterior visits, the lower visit shortens and the
upper visit lengthens. More explicitly, if $t_d$ and $t_u$ are the
descending and ascending lower crossings, then
\[
 (t_d)_y=-\frac{x_y(t_d;y)}{x_t(t_d;y)}>0,
 \qquad
 (t_u)_y=-\frac{x_y(t_u;y)}{x_t(t_u;y)}<0.
\]
Hence $s_y<0$. The corresponding calculation at the upper crossings gives
$v_y>0$. Therefore
\begin{equation}\label{eq:competition}
 T_y=-s_y-v_y
\end{equation}
is a sum of terms of opposite signs. The middle-zone residence time need
not be monotone even though the trajectories themselves are strictly ordered.

The disappearance of a lower visit has a useful local description.
Grazing singularities are a familiar organizing mechanism in
piecewise-smooth and impact dynamics \cite{diBernardo2008,Nordmark1991}.
Here the quantity of interest is the duration of the exterior visit,
whose local law can be derived directly.
\begin{proposition}\label{prop:grazing}
Suppose a local family of trajectories has a nondegenerate lower contact
at $(t_g,y_g)$, with
\[
 x(t_g;y_g)=-1,\qquad x_t(t_g;y_g)=0,\qquad
 \gamma=\mu\cos t_g>0.
\]
Assume this is the only lower contact in the neighborhood, and put
$z_g=x_y(t_g;y_g)>0$. The duration of the nearby lower visit satisfies
\begin{equation}\label{eq:sqrt}
 s(y)=2\sqrt{\frac{2z_g}{\gamma}}\sqrt{y_g-y}
       +o(\sqrt{y_g-y})\qquad(y\uparrow y_g).
\end{equation}
\end{proposition}
\begin{proof}
At the contact, the two one-sided second time derivatives agree because
the velocity vanishes and the vector field is continuous. They equal
$\mu\cos t_g$. Dependence on $y$ is differentiable at this isolated
contact: difference quotients satisfy the scalar variational equation
away from the switching times, and the bounded coefficients converge
almost everywhere. Thus, for $\delta=y_g-y$ and $\tau=t-t_g$,
\[
 x(t_g+\tau;y_g-\delta)+1
 =-z_g\delta+\frac\gamma2\tau^2+o(\delta+\tau^2).
\]
The neighboring roots have
$\tau_\pm=\pm\sqrt{2z_g\delta/\gamma}+o(\sqrt\delta)$.
Their separation gives \eqref{eq:sqrt}.
\end{proof}
The square-root dependence makes the shrinking lower visit particularly
sensitive near grazing, sharpening the competition in \eqref{eq:competition}.

\section{Locating the rational parameter example}\label{sec:construction}
We now describe the construction leading to the coefficients in
Theorem~\ref{thm:seven}. It begins with the critical residence time,
reduces the four-flight matching conditions to two duration variables,
and then converts the selected periodic profiles into rational shooting
intervals.

\subsection{A biased response and its critical duration}
Put $g(x)=Ax-(A+B)\sat(x)$ and introduce the auxiliary periodic problem
\begin{equation}\label{eq:biased}
 \dot x+g(x)=\mu\sin t+\lambda,\qquad x(t+2\pi)=x(t).
\end{equation}
The original equation is its slice $\lambda=0$. On a transverse
itinerary, let $X(t;y,\lambda)$ be the solution starting from $y$.
Differentiation with respect to the bias gives
\[
 X_\lambda(2\pi;y,\lambda)
 =\int_0^{2\pi}\exp\left(-\int_s^{2\pi}g'(X(\tau))\,d\tau\right)ds>0.
\]
Consequently the periodicity equation
$X(2\pi;y,\lambda)-y=0$ locally determines $\lambda$ as a function of $y$.
At a response critical point, its derivative vanishes precisely when
$X_y(2\pi;y,\lambda)=1$. At such a point the periodic variation in $y$
is strictly positive, so the mean
\[
 \xi=\frac1{2\pi}\int_0^{2\pi}X(t;y,\lambda(y))\,dt
\]
is a valid local coordinate. Write $\lambda=\Lambda(\xi,\mu)$ and
$z(t)=x_\xi(t;\xi,\mu)$. Then $z>0$ on a neighborhood of the critical
profile and $\dot z+g'(x)z=\Lambda_\xi$. Dividing by $z$ and integrating
over a period gives
\begin{equation}\label{eq:response-identity}
 \boxed{\Lambda_\xi J=(A+B)(\rho-T),\qquad
 J=\int_0^{2\pi}\frac{dt}{z(t)}>0,\qquad
 \rho=\frac{2\pi A}{A+B}.}
\end{equation}
Thus the response descends when $T>\rho$ and ascends when $T<\rho$.
At a critical profile, $T=\rho$. This reduces the search for changes in
the response slope to a question about flight durations.

\subsection{The four-crossing derivative test}
Take the switching levels in cyclic order $+1,-1,-1,+1$ at
$t_1<t_2<t_3<t_4$, and put
\begin{equation}\label{eq:construction-durations}
 r=t_2-t_1,\quad s=t_3-t_2,\quad
 u=t_4-t_3,\quad v=t_1+2\pi-t_4.
\end{equation}
The zones are middle, lower, middle, upper, so $T=r+u$.
At a critical profile, $u=\rho-r$ and $v=2\pi-\rho-s$.
Let $w_i=x_t(t_i)\ne0$ and $z_1=z(t_1)>0$.
Implicit differentiation of a crossing gives
$t_{i,\xi}=-z(t_i)/w_i$. Since $\Lambda_\xi=0$, the variation
is multiplied by $e^{Br}$, $e^{-As}$, and $e^{Bu}$ on the first
three flights. Therefore
\begin{equation}\label{eq:H-test}
 \boxed{T_\xi=z_1\mathcal H,\qquad
 \mathcal H=\frac1{w_1}-\frac{e^{Br}}{w_2}
       +\frac{e^{Br-As}}{w_3}
       -\frac{e^{B\rho-As}}{w_4}.}
\end{equation}
The sign of $\mathcal H$ measures the balance between the shortening
lower visit and the lengthening upper visit at a critical response.

\subsection{Symbolic reduction of the matching equations}
We integrate the four affine segments and eliminate their common bias
and forcing phase. The resulting system has only two duration variables.
Let the first crossing phase be $\phi=t_1$ and set
\[
 q_0=0,\quad q_1=r,\quad q_2=r+s,\quad
 q_3=\rho+s,\quad q_4=2\pi.
\]
The endpoint levels are
$(b_0,b_1,b_2,b_3,b_4)=(1,-1,-1,1,1)$.
For the four flights use
\[
 (c_j,d_j)=(-B,0),\ (A,A+B),\ (-B,0),\ (A,-A-B).
\]
Put $L_j=q_j-q_{j-1}$, $E_j=e^{-c_jL_j}$, and
\begin{align}
 \alpha_j&=\frac{1-E_j}{c_j},\notag\\
 \beta_j&=\frac{c_j\cos q_j+\sin q_j
       -E_j(c_j\cos q_{j-1}+\sin q_{j-1})}{1+c_j^2},\notag\\
 \gamma_j&=\frac{c_j\sin q_j-\cos q_j
       -E_j(c_j\sin q_{j-1}-\cos q_{j-1})}{1+c_j^2},\label{eq:matching-coefficients}\\
 k_j&=b_j-E_jb_{j-1}+d_j\alpha_j.\notag
\end{align}
With $U=\mu\sin\phi$ and $V=\mu\cos\phi$, the four endpoint
conditions become the linear equations
\begin{equation}\label{eq:linear-matching}
 \alpha_j\lambda+\beta_jU+\gamma_jV=k_j,
 \qquad j=1,2,3,4.
\end{equation}
The angle-addition formulas separate $\phi$ from the relative times $q_j$.
Integrating $x'+c_jx+d_j=\lambda+\mu\sin t$ gives
$b_j=E_jb_{j-1}+\alpha_j(\lambda-d_j)+\beta_jU+\gamma_jV$.
Thus every four-crossing critical profile satisfies the compatibility equation
\begin{equation}\label{eq:compatibility}
 \mathcal F(A,B,r,s):=
 \det\begin{pmatrix}
 \alpha_1&\beta_1&\gamma_1&k_1\\
 \alpha_2&\beta_2&\gamma_2&k_2\\
 \alpha_3&\beta_3&\gamma_3&k_3\\
 \alpha_4&\beta_4&\gamma_4&k_4
 \end{pmatrix}=0.
\end{equation}
Choose three coefficient rows with determinant $\Delta\ne0$, and let
$\Delta_\lambda,\Delta_U,\Delta_V$ be the determinants obtained by
replacing the corresponding column with the right-hand side.
Cramer's rule gives
\[
 \lambda=\frac{\Delta_\lambda}{\Delta},\qquad
 U=\frac{\Delta_U}{\Delta},\qquad V=\frac{\Delta_V}{\Delta}.
\]
The fourth row gives \eqref{eq:compatibility}, and a prescribed amplitude
gives the second equation
\begin{equation}\label{eq:symbolic-amplitude}
 \mathcal G(r,s;\mu):=\Delta_U^2+\Delta_V^2-\mu^2\Delta^2=0.
\end{equation}
Thus symbolic elimination reduces the critical matching problem to
$\mathcal F=\mathcal G=0$. The coefficients are explicit expressions
in exponentials, sines, and cosines of the two residence variables.
The phase and amplitude are reconstructed by
\begin{equation}\label{eq:amplitude-recovery}
 \mu=\sqrt{U^2+V^2},\qquad
 \sin\phi=U/\mu,\quad\cos\phi=V/\mu.
\end{equation}
Physical profiles satisfy $0<r<\rho$, $0<s<2\pi-\rho$, and the prescribed
crossing directions. The weighted-velocity identity \eqref{eq:shape}
and the speeds $w_i$ determine the first encounters on each flight.

\subsection{Selecting the slopes, amplitude, and shooting intervals}
The derivative test \eqref{eq:H-test} led us to profiles with strong
exterior contraction. The simple rational slopes
\[
 A=9,\qquad B=\frac32,\qquad \rho=\frac{12\pi}{7}
\]
give such a critical branch. Continuing the physical solutions of
\eqref{eq:compatibility} and reconstructing the amplitude by
\eqref{eq:amplitude-recovery} gives a critical-amplitude curve with a
local minimum near $1.813584$. The symmetric critical profile occurs
near $1.81363456$. We choose a fixed amplitude between these two values.
The latter value comes from the symmetric matching conditions
$r=u=\rho/2$, $s=v=(2\pi-\rho)/2$, and $\lambda=0$ in
\eqref{eq:linear-matching}. A stationary amplitude along a regular
compatibility chart is located by the two equations
\[
 \mathcal F=0,\qquad
 \mathcal F_r\mu_s-\mathcal F_s\mu_r=0,
\]
where $\mu(r,s)$ is reconstructed from the chosen coefficient minor.
This gives the rational choice
\begin{equation}\label{eq:selected-amplitude}
 \mu=\frac{181363}{100000}=1.81363.
\end{equation}
At the selected amplitude, the two positive-mean three-zone critical
matches have approximately
\[
\begin{array}{c|rrr}
 &r&s&\mathcal H\\\hline
\text{first}&2.73100378&0.40448498&-3.31230\\
\text{second}&2.90517603&0.18921685&80.43176.
\end{array}
\]
The opposite signs in \eqref{eq:H-test} explain the extra reversal of
the residence time on this branch.

At this amplitude, the continued biased response has the three turning
profiles shown in Table~\ref{tab:construction-turns}. The first two
visit all three zones; the third lies on the right two-zone branch after
the lower exterior visit disappears.
\begin{table}[ht]
\centering
\caption{Rounded continuation values guiding the construction.}\label{tab:construction-turns}
\begin{tabular}{@{}rrl@{}}\toprule
Mean $\xi$ & Bias $\Lambda$ & Profile\\\midrule
$0.00427486$ & $-5.84\cdot10^{-7}$ & three-zone minimum\\
$0.02615026$ & $1.20\cdot10^{-5}$ & three-zone maximum\\
$0.04964750$ & $-2.30\cdot10^{-5}$ & two-zone minimum\\\bottomrule
\end{tabular}
\end{table}
Starting at the symmetric profile $(\xi,\Lambda)=(0,0)$, the response
descends, rises, descends, and rises again. This is the geometry that
suggests three positive-mean zeros. On a right two-zone branch,
strict order makes the upper exterior visit longer as the mean rises,
so $T=2\pi-v$ decreases. In the three-zone part,
the shrinking lower visit adds the opposing contribution in
\eqref{eq:competition}. Near its loss, the square-root law
\eqref{eq:sqrt} accounts for the sharp residence-time turn.

\begin{figure}[ht]
\centering
\begin{tikzpicture}
\begin{groupplot}[group style={group size=2 by 1,horizontal sep=1.6cm},
 width=.43\linewidth,height=5cm,xmin=0,xmax=.08,
 xlabel={Mean $\xi$},grid=major,grid style={gray!20},
 tick label style={font=\scriptsize},label style={font=\small},
 xtick={0,.02,.04,.06,.08},scaled x ticks=false,
 xticklabel style={/pgf/number format/fixed,/pgf/number format/precision=2}]
\nextgroupplot[ylabel={$10^5\Lambda$}]
\addplot[blue!65!black,thick] table[x=xi,y=scaled_lambda,col sep=comma]{figures/response_data.csv};
\addplot[dashed] coordinates {(0,0)(.08,0)};
\addplot[only marks,mark=x,red!70!black] coordinates
 {(.007659607,0)(.037207634,0)(.065677290,0)};
\nextgroupplot[ylabel={$10^3(T-\rho)$}]
\addplot[teal!70!black,thick] table[x=xi,y expr=1000*\thisrow{T_minus_rho},col sep=comma]{figures/response_data.csv};
\addplot[dashed] coordinates {(0,0)(.08,0)};
\end{groupplot}
\end{tikzpicture}
\caption{Sampled response and residence-time profiles at
$(A,B,\mu)=(9,3/2,1.81363)$. Crosses mark the three positive-mean
responses.}\label{fig:construction}
\end{figure}
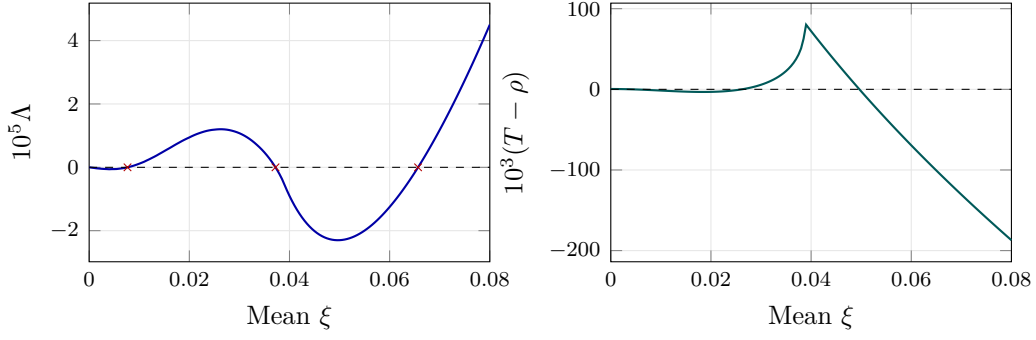

Finally, solving the unbiased affine matching equations gives approximate
shooting values
\[
 -0.5692039544,\qquad -0.5567601961,\qquad -0.5551559903.
\]
These guide the rational endpoint choices in \eqref{eq:intervals}.
The second interval is narrow because its trajectory nearly grazes the
lower level; the first and third can be bracketed more broadly.
Section~4 derives the displacement signs at these rational endpoints
and uses scalar order and reflection to prove the seven-cycle theorem.

\section{Proof of the seven-cycle theorem}
From now on fix $A=9$, $B=3/2$, and $\mu=181363/100000$.
In a zone where
$\dot x+cx+d=\mu\sin t$, the exact flow is
\begin{align}\label{eq:flight}
 \Phi_{c,d}(t;t_0,y_0)
 ={}&e^{-c(t-t_0)}y_0-\frac dc\bigl(1-e^{-c(t-t_0)}\bigr)\\
 &+\frac{\mu}{1+c^2}
 \left[c\sin t-\cos t
 -e^{-c(t-t_0)}(c\sin t_0-\cos t_0)\right].\notag
\end{align}
The pairs $(c,d)$ in the middle, lower, and upper zones are
$(-3/2,0)$, $(9,21/2)$, and $(9,-21/2)$, respectively.
We join these flows at their first encounters with the switching levels.

\subsection{Three intervals with strict displacement signs}
\begin{lemma}\label{lem:bounds}
The displacement $D$ has the bounds in Table~\ref{tab:signs}.
On each of
\begin{equation}\label{eq:intervals}
 I_1=[-0.571,-0.567],\quad
 I_2=[-0.55677,-0.55675],\quad
 I_3=[-0.55517,-0.55514],
\end{equation}
all switching encounters are transverse. The itineraries are
middle--lower--middle--upper--middle on $I_1,I_2$ and
middle--upper--middle on $I_3$.
\end{lemma}
\begin{table}[ht]
\centering
\caption{Rational bounds for the return displacement. All decimal endpoints are exact rationals.}\label{tab:signs}
\begin{tabular}{@{}rc@{}}\toprule
$y$ & Interval containing $D(y)$\\\midrule
$-0.571$ & $[1.13334,1.13357]\cdot10^{-5}$\\
$-0.567$ & $[-2.57304,-2.57281]\cdot10^{-5}$\\
$-0.55677$ & $[-4.97010,-4.96788]\cdot10^{-6}$\\
$-0.55675$ & $[5.38442,5.38665]\cdot10^{-6}$\\
$-0.55517$ & $[9.31073,9.31075]\cdot10^{-6}$\\
$-0.55514$ & $[-1.07074,-1.07072]\cdot10^{-5}$\\\bottomrule
\end{tabular}
\end{table}
\begin{proof}
Appendix~\ref{app:bounds} gives the rational estimates for the six
endpoint trajectories. The weighted velocity determines their first
encounters and the persistence of their itineraries between endpoints.

Inside an affine zone, the weighted velocity obeys
\begin{equation}\label{eq:shape}
 \frac d{dt}\left(e^{c(t-t_0)}\dot x(t)\right)
       =\mu e^{c(t-t_0)}\cos t.
\end{equation}
It increases up to $\pi/2$, decreases from $\pi/2$ to $3\pi/2$, and
increases from $3\pi/2$ to $2\pi$. For the four-crossing endpoint
trajectories, the first two events precede $\pi/2$, the third follows
$\pi$, and the fourth precedes $3\pi/2$. Their velocities have signs
$(-,+,+,-)$. Consequently the first middle flight decreases to its first
lower encounter; the lower flight has a single minimum; the next middle
flight increases to its first upper encounter; and the upper flight has
a single maximum. The last middle flight has negative velocity at both
ends and remains decreasing. Its final value is greater than $-1$.
For the two-crossing trajectories, the initial middle flight has one
minimum, whose value is strictly greater than $-1$, and then reaches the
upper level. The upper and terminal middle flights have the same shape
as above. These observations exclude any omitted first encounter.

A tangency at $\ell\in\{-1,1\}$ can occur only when
\begin{equation}\label{eq:phase}
 \mu\sin t+\frac32\ell=0.
\end{equation}
There are four candidate phases per period. In their chronological order,
the signs of $x(t;y)-\ell$ at both endpoints of each interval are
\begin{center}
\begin{tabular}{@{}lcccc@{}}\toprule
 & lower 1 & lower 2 & upper 1 & upper 2\\\midrule
$I_1$ & $-$ & $+$ & $+$ & $-$\\
$I_2$ & $-$ & $+$ & $+$ & $-$\\
$I_3$ & $+$ & $+$ & $+$ & $-$\\\bottomrule
\end{tabular}
\end{center}
Each absolute gap exceeds $3\cdot10^{-4}$ by the rational estimates
in the appendix. Strict scalar order keeps these signs unchanged
throughout each $I_j$. Thus no trajectory in these intervals can graze
a switching level. The initial and final values remain inside $(-1,1)$,
so no encounter can enter or leave through a temporal endpoint.
Transverse encounters persist by the implicit function theorem; hence
their number and ordering are constant on each connected interval.
Finally, substituting the event enclosures in \eqref{eq:flight} gives
Table~\ref{tab:signs}.
\end{proof}

The continuous function $D$ changes sign on every $I_j$, so it has a zero
on each interval. Each zero is isolated. Indeed, transverse crossing
times depend real analytically on $y$ by \eqref{eq:flight} and the
implicit function theorem. Hence $D$ is real analytic in a neighborhood
of each $I_j$. Its nonzero endpoint values exclude an identically
vanishing analytic germ. We obtain three distinct isolated periodic solutions.

\subsection{The central orbit and reflection}
Define
\[
 Q(y)=x(\pi;y)+y.
\]
Strict scalar order makes $Q$ strictly increasing. Variation of constants
in the form $\dot x=-9x+O(1)$ gives
$x(\pi;y)=e^{-9\pi}y+O(1)$ as $y\to\pm\infty$, uniformly in $y$.
Thus $Q$ has exactly one zero $y_c$.
Oddness of the spatial vector field and
$\sin(t+\pi)=-\sin t$ then imply
\begin{equation}\label{eq:anti}
 x_c(t+\pi)=-x_c(t),\qquad x_c(t)=x(t;y_c).
\end{equation}
In particular, $x_c$ is $2\pi$-periodic.

The additional rational bounds in Appendix~\ref{app:center} give
\begin{equation}\label{eq:center}
 -0.579022471<y_c<-0.579022470
\end{equation}
and show that its half-period lower-zone residence $s_c$ obeys
\begin{equation}\label{eq:sc}
 s_c<0.448589954<\frac\pi7.
\end{equation}
By \eqref{eq:anti}, its full-period exterior residence is $2s_c$.
Equation~\eqref{eq:multiplier} yields
\[
 \log P'(y_c)=3\pi-21s_c>0.
\]
Therefore the central orbit is isolated.

Every $I_j$ lies above $y_c$. Strict scalar order puts all three selected
periodic solutions above $x_c$ at every phase. The map
\[
 \mathcal Sx(t)=-x(t+\pi)
\]
is an involution on periodic solutions. It sends each of these three
solutions to one below $x_c$, preserves distinctness, and preserves
isolation. Together with $x_c$, this gives $3+1+3=7$ distinct isolated
periodic orbits and proves Theorem~\ref{thm:seven}.\qed

\appendix
\section{Finite rational estimates for the affine flows}\label{app:bounds}
We give the finite inequalities used in Lemma~\ref{lem:bounds}.
This appendix specifies the rational inputs, the elementary function bounds,
and the finite arithmetic yielding the displayed strict inequalities.

\subsection{Elementary functions and rational operations}
For $|u|\le10$, use
\begin{align*}
 E(u)&=\sum_{k=0}^{160}\frac{u^k}{k!},\\
 S(u)&=\sum_{k=0}^{80}\frac{(-1)^ku^{2k+1}}{(2k+1)!},\qquad
 C(u)=\sum_{k=0}^{80}\frac{(-1)^ku^{2k}}{(2k)!}.
\end{align*}
Each polynomial differs from the corresponding exponential, sine, or
cosine by less than $10^{-95}$. Taylor's theorem gives the common bound
\[
 \frac{3^{10}10^{161}}{161!}<10^{-95};
\]
for the sine polynomial, whose degree is $161$, the degree-$162$
remainder is smaller still. For rational intervals, evaluate each sum
by its term recurrence, enclosing every operation. For example,
$e_0=1$, $e_{k+1}=e_ku/(k+1)$, and
$s_0=u$, $s_{k+1}=-s_ku^2/((2k+2)(2k+3))$.
Add $[-10^{-95},10^{-95}]$ to the final polynomial interval.

For interval arithmetic, use the rational lattice $10^{-80}\mathbb Z$.
After each operation replace a lower endpoint $r$ by
$10^{-80}\lfloor10^{80}r\rfloor$ and an upper endpoint by
$10^{-80}\lceil10^{80}r\rceil$. Use the minimum and maximum of the four
endpoint products for multiplication, and reciprocal endpoints for
division by an interval excluding zero. Unary negation is exact.
Thus all subsequent interval endpoints are rational and include every
rounding enlargement.

To enclose $\pi$, let
\[
 a_n(k)=\sum_{j=0}^{n-1}\frac{(-1)^j}{(2j+1)k^{2j+1}}.
\]
For $n=65$, the alternating-series bounds give
\[
 a_{65}(k)-\frac1{131k^{131}}
 <\arctan(1/k)<a_{65}(k),\qquad k=5,239.
\]
Apply $\pi=16\arctan(1/5)-4\arctan(1/239)$ with these rational bounds.
All arguments in the following affine-flow evaluations have modulus
less than $10$.

\subsection{Switching brackets and endpoint signs}
Every decimal in Table~\ref{tab:events} denotes an exact rational.
For four-event rows put
\[
 h_1=10^{-27},\quad h_2=10^{-22},\quad
 h_3=10^{-19},\quad h_4=10^{-12};
\]
for two-event rows put $h_1=10^{-27}$, $h_2=10^{-19}$.
The event interval is $H_j=[a_j-h_j,a_j+h_j]$.
Start with $(H_0,X_0)=(\{0\},\{y\})$. At each step evaluate
\[
 F_j^\pm=\Phi_{c_j,d_j}(a_j\pm h_j;H_{j-1},X_{j-1})-\ell_j,
 \qquad V_j=\mu\sin H_j+\tfrac32\ell_j.
\]
Here $\Phi$ on intervals means substitution in \eqref{eq:flight} using
the rational operations above. For four-event rows use the incoming
zones middle, lower, middle, upper and levels $(-1,-1,1,1)$.
For two-event rows use middle, upper and levels $(1,1)$.
Their finite rational inequalities are
\[
\begin{array}{c|cccc}
j&1&2&3&4\\\hline
\text{four events: }(\operatorname{sgn}F_j^-,\operatorname{sgn}F_j^+,
 \operatorname{sgn}V_j)&(+,-,-)&(-,+,+)&(-,+,+)&(+,-,-)
\end{array}
\]
and $(-,+,+),(+,-,-)$ for the two-event rows.
Set $X_j=\{\ell_j\}$ after each crossing. The terminal middle flight
$\Phi_{-3/2,0}(2\pi;H_k,X_k)-y$ then has the enclosures in
Table~\ref{tab:signs}. Its final state is in $(-1,1)$ and its final
velocity is negative in every row.
{\small
\begin{longtable}{@{}rrl@{}}
\caption{Exact rational centers of the switching brackets.}\label{tab:events}\\
\toprule $y$ & $j$ & $a_j$\\\midrule\endfirsthead
\toprule $y$ & $j$ & $a_j$\\\midrule\endhead
$-0.571$ & $1$ & $0.690900167350795596241598610060$\\
$-0.571$ & $2$ & $1.078135763256395127520565407613$\\
$-0.571$ & $3$ & $3.717538493711361911132150876164$\\
$-0.571$ & $4$ & $4.228305204642255023584881326373$\\
\addlinespace[2pt]
$-0.567$ & $1$ & $0.727167718353076856080532101227$\\
$-0.567$ & $2$ & $1.073698524509500870150051069731$\\
$-0.567$ & $3$ & $3.677514848910453028389655375076$\\
$-0.567$ & $4$ & $4.230022149628273629974937792239$\\
\addlinespace[2pt]
$-0.55677$ & $1$ & $0.943026106870525525503040801495$\\
$-0.55677$ & $2$ & $1.000234039978349548083465355225$\\
$-0.55677$ & $3$ & $3.432412794615143204771577791489$\\
$-0.55677$ & $4$ & $4.234378779776890243529064894432$\\
\addlinespace[2pt]
$-0.55675$ & $1$ & $0.945709233396853423992784093457$\\
$-0.55675$ & $2$ & $0.998251562666127082230286839438$\\
$-0.55675$ & $3$ & $3.430444659672579159693512100466$\\
$-0.55675$ & $4$ & $4.234391508369720611532563245294$\\
\addlinespace[2pt]
$-0.55517$ & $1$ & $3.234089094501122628580828533598$\\
$-0.55517$ & $2$ & $4.235054342221641269610541466055$\\
\addlinespace[2pt]
$-0.55514$ & $1$ & $3.231221607277306682192172796767$\\
$-0.55514$ & $2$ & $4.235058511047931109063368566348$\\
\addlinespace[2pt]
\bottomrule
\end{longtable}}
For $y=-571/1000$, the first
crossing equation reduces to
\[
 F_1(t)=1-\frac{1053}{81250}e^{3t/2}
 -\frac{181363}{325000}\left(\frac32\sin t+\cos t\right)=0.
\]
Insert $t=a_1\pm10^{-27}$ from Table~\ref{tab:events}, bound the
three elementary functions by the polynomials above, and collect the
rational terms. This gives
\[
 F_1(a_1-10^{-27})\in[3.4452,3.4453]\cdot10^{-28},\qquad
 F_1(a_1+10^{-27})\in[-3.4408,-3.4407]\cdot10^{-28}.
\]
The remaining substitutions for this shooting value are displayed below.
The intervals in the second and third columns are to be multiplied by
the scale in the last column.
\begin{center}
\begin{tabular}{@{}rlll@{}}\toprule
$j$ & $F_j^-$ & $F_j^+$ & Scale\\\midrule
2 & $[-10.464,-9.126]$ & $[9.126,10.464]$ & $10^{-24}$\\
3 & $[-7.796,-2.449]$ & $[2.449,7.796]$ & $10^{-20}$\\
4 & $[1.0504,1.0545]$ & $[-1.0545,-1.0504]$ & $10^{-13}$\\\bottomrule
\end{tabular}
\end{center}
The crossing-speed bounds, in order, are
\[
 [-0.34430,-0.34429],\quad[0.09794,0.09795],\quad
 [0.51224,0.51225],\quad[-0.10525,-0.10524].
\]
After the fourth crossing the terminal middle flight gives
\[
 P(-0.571)\in[-0.57098866660,-0.57098866437].
\]
Subtracting $-0.571$ proves the first row of Table~\ref{tab:signs}.
The other five rows follow by the same substitutions with their listed
rational centers. In each case the task is a finite comparison of
rational numbers, using the stated Taylor remainder bounds.

For the two-event rows, the initial middle-flight minimum is bracketed
by the following centers with radius $10^{-22}$:
\[
\begin{array}{c|l}
y&m_y\\\hline
-0.55517&0.964606148956481813734238900023\\
-0.55514&0.964421114280833615237799576884.
\end{array}
\]
The velocity at the two bracket endpoints has signs $(-,+)$, and
\[
 \Phi_{-3/2,0}([m_y-10^{-22},m_y+10^{-22}];0,y)>-1.
\]
Together with \eqref{eq:shape}, these inequalities justify the first-hit
interpretation of all segments.

\subsection{Excluding grazing on the complete shooting intervals}
Bracket the four roots of \eqref{eq:phase} using the exact centers
\[
\begin{array}{c|l}
\ell&a\\\hline
-1&0.973875926376460045824009680088\\
-1&2.167716727213333192638633703191\\
 1&4.115468579966253284286653063368\\
 1&5.309309380803126431101277086471
\end{array}
\]
and radius $10^{-25}$. The level velocity has opposite signs at the
two endpoints of each phase bracket. Use the already enclosed switching
times to select the affine segment containing that bracket and evaluate
$\Phi_{c,d}(H;H_j,X_j)-\ell$ there. The gap magnitudes exceed
$3\cdot10^{-4}$ for all six shooting endpoints, with signs as displayed
in the proof of Lemma~\ref{lem:bounds}. Thus the same strict signs hold
at the exact roots of \eqref{eq:phase}.

\subsection{The antisymmetric orbit}\label{app:center}
For the two rational initial values in Table~\ref{tab:center-events},
join middle, lower, and middle flights up to $\pi$. The two event
radii are $10^{-27}$ and $10^{-22}$, respectively. Both levels are $-1$;
their crossing directions are downward and upward.
{\small
\begin{longtable}{@{}rrl@{}}
\caption{Half-period event centers for the central bracket.}\label{tab:center-events}\\
\toprule $y$ & $j$ & $a_j$\\\midrule\endfirsthead
\toprule $y$ & $j$ & $a_j$\\\midrule\endhead
$-0.579022471$ & $1$ & $0.634581372314921880292542675282205728$\\
$-0.579022471$ & $2$ & $1.083171326123873562346343278450922288$\\
$-0.579022470$ & $1$ & $0.634581378413117209211221821647541674$\\
$-0.579022470$ & $2$ & $1.083171325676822885739904090457136619$\\
\bottomrule\end{longtable}}
The same rational substitutions give
\begin{align*}
 Q(-0.579022471)&\in[-5.768,-5.767]\cdot10^{-10},\\
 Q(-0.579022470)&\in[1.4254,1.4255]\cdot10^{-9}.
\end{align*}
Strict monotonicity of $Q$ proves \eqref{eq:center}.
For the complete bracket between these initial values, no upper crossing
occurs on $[0,\pi]$: at the upper level the velocity is strictly positive,
so an orbit could not return to its terminal value below $1$ after an
upper crossing. The lower level has at most one downward and one upward
crossing on this half-period, since its threshold velocity is negative,
then positive, then negative; a further downward crossing would leave
the orbit below $-1$ at $\pi$, contrary to the endpoint bounds and order.
At the first lower candidate tangency phase both bracket trajectories
are strictly below $-1$, and at the second both are above $-1$.
These signs follow from the same flight substitutions, and exclude a
grazing throughout the central bracket.

Increasing $y$ moves the downward crossing to the right and the upward
crossing to the left, by the crossing derivative formula in Section~2.
Thus the lower-zone duration is largest at $y=-0.579022471$. Its upper bound is
\[
\begin{split}
 s_c<&\;1.083171326123873562346343278450922288+10^{-22}\\
 &-0.634581372314921880292542675282205728+10^{-27}\\
 <&\;0.448589954<\pi/7.
\end{split}
\]
The last comparison follows already from the elementary bound
$\pi>3.14159$. This proves \eqref{eq:sc} and completes the finite estimates.

\section*{Use of AI tools}
AI assistance was used to write the programs for parameter searches and
related computations, and to help draft the manuscript. The authors
checked the mathematical arguments and proofs presented in this paper.

\end{document}